\documentclass[conference]{IEEEtran}
\usepackage{amsfonts}
\usepackage{dsfont}
\usepackage{amssymb}
\usepackage{graphicx}
\usepackage{subfigure}
\usepackage{enumerate}
\usepackage{amsmath}
\usepackage{color}
\usepackage{amsthm}
\usepackage{amsmath}
\usepackage{algorithm}
\usepackage{algpseudocode}
\usepackage[bookmarks=false]{hyperref}
\hypersetup{hidelinks}
\usepackage{breakurl}
\usepackage{cite}
\usepackage{epstopdf}

\DeclareMathOperator*{\argmin}{arg\,min}

\newcommand{\bp}{\begin{proof} \small }
\newcommand{\ep}{\end{proof} \normalsize}
\newcommand{\epx}{\end{proof} \small}
\newcommand{\bpa}{\begin{proofappx} \footnotesize }
\newcommand{\epa}{\end{proofappx} \small }
\newtheorem{theorem}{Theorem}

\newtheorem*{theorem*}{Theorem}
\newtheorem*{proposition*}{Proposition}
\newtheorem*{corollary*}{Corollary}
\newtheorem*{lemma*}{Lemma}
\newtheorem*{assumption*}{Assumption}
\newtheorem*{definition*}{Definition}
\newtheorem*{claim*}{Claim}

\newcommand{\be}{\begin{equation}}
\newcommand{\ee}{\end{equation}}
\newcommand{\bs}{\begin{subequations}}
\newcommand{\es}{\end{subequations}}
\newcommand{\bq}{\begin{eqnarray}}
\newcommand{\eq}{\end{eqnarray}}
\newcommand{\bqn}{\begin{eqnarray*}}
\newcommand{\eqn}{\end{eqnarray*}}

\newcommand{\ba}{\left[ \begin{array}}
\newcommand{\ea}{\\ \end{array} \right]}
\newcommand{\ben}{\begin{enumerate}}
\newcommand{\een}{\end{enumerate}}

\def\d{{\boldsymbol{d}}}

\def\real{{\mathchoice%
{\hbox{\rm\setbox1=\hbox{I}\copy1\kern-.45\wd1 R}}
{\hbox{\rm\setbox1=\hbox{I}\copy1\kern-.45\wd1 R}}
{\hbox{\scriptsize\rm\setbox1=\hbox{I}\copy1\kern-.45\wd1 R}}
{\hbox{\scriptsize\rm\setbox1=\hbox{I}\copy1\kern-.45\wd1 R}}}}

\def\Zint{{\mathchoice{\setbox1=\hbox{\sf Z}\copy1\kern-.75\wd1\box1}
{\setbox1=\hbox{\sf Z}\copy1\kern-.75\wd1\box1}
{\setbox1=\hbox{\scriptsize\sf Z}\copy1\kern-.75\wd1\box1}
{\setbox1=\hbox{\scriptsize\sf Z}\copy1\kern-.75\wd1\box1}}}
\newcommand{\complex}{ \hbox{\rm C\kern-0.45em\rule[.07em]{.02em}{.58em}%
\kern 0.43em}}

\begin{document}
	%
\title{Task Replication for Vehicular Edge Computing: A Combinatorial Multi-Armed Bandit based Approach}
	
	\author{\IEEEauthorblockN{Yuxuan Sun$^*$, Jinhui Song$^*$, Sheng Zhou$^*$, Xueying Guo$^\dagger$, Zhisheng Niu$^*$}
		\IEEEauthorblockA{$^*$Beijing National Research Center for Information Science and Technology\\
			Department of Electronic Engineering, Tsinghua University, Beijing, China\\
            Email: \{sunyx15, sjh14\}@mails.tsinghua.edu.cn, \{sheng.zhou, niuzhs\}@tsinghua.edu.cn	\\
        $^\dagger$Department of Computer Science, University of California, Davis, CA, USA\\
        Email: guoxueying@outlook.com}}
        

	\maketitle

	\begin{abstract}
In vehicular edge computing (VEC) system, some vehicles with surplus computing resources can provide computation task offloading opportunities for other vehicles or pedestrians. However, vehicular network is highly dynamic, with fast varying channel states and computation loads. 
These dynamics are difficult to model or to predict, but they have major impact on the quality of service (QoS) of task offloading, including delay performance and service reliability.
Meanwhile, the computing resources in VEC are often redundant due to the high density of vehicles.
To improve the QoS of VEC and exploit the abundant computing resources on vehicles, we propose a learning-based task replication algorithm (LTRA) based on combinatorial multi-armed bandit (CMAB) theory, in order to minimize the average offloading delay.
LTRA enables multiple vehicles to process the replicas of the same task simultaneously, and vehicles that require computing services can learn the delay performance of other vehicles while offloading tasks.
We take the occurrence time of vehicles into consideration, and redesign the utility function of existing CMAB algorithm, so that LTRA can adapt to the time varying network topology of VEC.
We use a realistic highway scenario to evaluate the delay performance and service reliability of LTRA through simulations, and show that compared with single task offloading, LTRA can improve the task completion ratio with deadline $0.6\mathrm{s}$ from $80\%$ to $98\%$.
	\end{abstract}
	

	%
	\IEEEpeerreviewmaketitle
	\section{Introduction}
	To support autonomous driving and various kinds of on-board infotainment services, future vehicles will possess strong computing capabilities. It is predicted that each vehicle needs about $10^6$ dhrystone million instructions executed per second (DMIPS) \cite{intel} computing power to enable self-driving.
	To deliver safety messages or disseminate infotainment contents, vehicles also need to communicate with other vehicles or infrastructures through vehicle-to-vehicle (V2V) and vehicle-to-infrastructure (V2I) communication protocols \cite{zhangjsac}, such as dedicated short-range communication (DSRC) \cite{kenney2011dsrc} protocol and LTE-V \cite{ltev}. 
	Consequently, vehicles will be connected with each other and abundant in computing resources in the future.
	
	To improve the utilization of vehicular computing resources, the concept of vehicle-as-an-infrastructure has been proposed \cite{hou2016vfc}, where vehicles can contribute their surplus computing resources to the network, forming Vehicular Edge Computing (VEC) systems \cite{abdel2015vehicle,bitam2015vanet,choo2017sdvc}.	
	VEC has huge potential to enhance edge intelligence, and can enable a variety of emerging applications that require intensive computing.
	Typical use cases include safety-related cooperative collision avoidance and collective environment perception for autonomous driving\cite{5gv2x,zhangshan}, vehicular crowdsensing for road monitoring and parking navigation \cite{vehcrowd}, and entertainments such as virtual reality, augmented reality and cloud gaming for passengers \cite{mao2017mobile}.	
	
	In VEC, computation tasks are generated by on-board driving systems, passengers or pedestrians, and can possibly be executed by vehicles through \emph{task offloading}. In this context, vehicles who provide cloud execution are called service vehicles (SeVs), while vehicles that require task offloading are called task vehicles (TaVs).
	In the literature, a semi-Markov decision process based formulation for centralized task offloading is given in \cite{zheng2015smdp}, in order to minimize the average utility related to delay and energy cost.
	However, centralized task scheduling requires to collect the complete state information of vehicles frequently, and the proposed algorithm is highly complex to run. 
	An alternative way is to offload tasks in a distributed manner, i.e., each TaV makes task offloading decisions individually \cite{fengave}. In this case, TaV may not be able to obtain the global state information of channel states and computation loads of all available SeVs, which can be learned while offloading tasks based on multi-armed bandit (MAB) theory, as shown in our previous work \cite{Sun2018ICC}.
	

	Compared with mobile edge computing (MEC) \cite{mao2017mobile}, in which computing resources are deployed at static base stations, VEC has two major  differences.
	On the one hand, vehicles move fast, making the network topology and wireless channels vary rapidly over time. 
	On the other hand, the density of SeVs is much higher than static edge clouds, and thus the computing resources of VEC are more redundant than MEC.
	

	To further improve the delay performance and service reliability in VEC system, while exploiting the redundancy of computing resources, \emph{task replication} is a promising method, in which task replicas are offloaded to multiple SeVs at the same time and executed independently. Once one of these SeVs transmits back the result, the task is completed.
	The basic idea of task replication is to exchange the redundancy of computing resources for QoS improvement.
	A centralized task replication algorithm is proposed in \cite{Jiang2017IoT}, in order to maximize the probability of completing a task before a given deadline. However, the optimal task assignment policy is derived under the assumption that the arrival of SeV follows Poisson process, which may not be the realistic vehicle mobility model.
	
	In this paper, we propose a learning-based task replication algorithm (LTRA) based on combinatorial multi-armed bandit (CMAB) theory \cite{chen2013cmab}. 
	To be specific, we first propose a distributed task replication framework, in order to minimize the average offloading delay. 
	Based on CMAB theory, we then design LTRA to deal with the challenge that TaV lacks global state information of channel states and computation loads of candidate SeVs, and characterize the upper bound of its learning regret.
	We simulate the traffic in a realistic highway scenario via traffic simulator Simulation for Urban MObility (SUMO), and compare LTRA with our previously proposed single offloading algorithm in \cite{Sun2018ICC}. Results show that both the average offloading delay and service reliability can be improved substantially through task replication.
	
%
%

	The rest of this paper is organized as follows.
	In Section \ref{sys}, we present the system model and problem formulation.
	The task replication algorithm is then proposed in Section \ref{algo}, followed by the performance analysis in Section \ref{per}.
	Simulation results are shown in Section \ref{sim}.
	And finally, we conclude the work in Section \ref{con}.

	\section{System Model and Problem Formulation} \label{sys}
	
	\subsection{System Overview}

	In the VEC system, moving vehicles are classified into two categories according to their roles in task offloading: TaVs and SeVs. TaVs are the vehicles who generate computation tasks that require cloud execution, while SeVs can share their surplus computing resources and execute these computation tasks.  
	Note that each vehicle may be either a TaV or a SeV, and its role can change over time, depending on whether it has surplus computing resources to share.
	
	Each TaV first discovers the SeVs within its communication range for task offloading. In order to maintain a relatively long contact duration, each TaV only selects its neighboring SeVs with the same moving direction as candidates. Such information can be acquired from V2V communication protocols such as beaconing messages in DSRC \cite{kenney2011dsrc}.
	Moreover, we do not make any assumptions on the mobility model of vehicles.
	
	We adopt \emph{task replication} technique to improve the reliability, and thus the delay of task offloading. 
	Besides, we consider distributed offloading in this work: each TaV makes the task offloading decision on which SeVs should be selected to serve each task independently, without inter-vehicle coordinations.
	An exemplary task replication in VEC system is shown in Fig. \ref{system}, where TaV 1 finds SeVs 1-3 as candidates, and decides to offload the current task replicas to SeV 1 and SeV 3.

	\begin{figure}  [!t]
		\centering
		\includegraphics[width=0.5\textwidth]{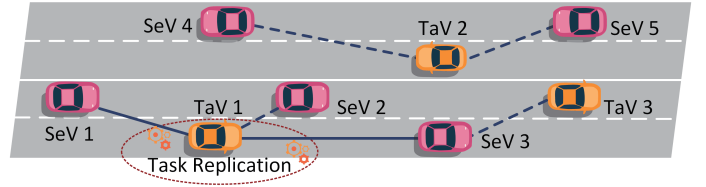}\\
		\caption{Task replication in VEC system.}\label{system}
	\end{figure}

    \subsection{Task Offloading Procedure}
    Since tasks are offloaded in a distributed manner, we then focus on a single TaV of interest and design the corresponding task offloading algorithm.
    Consider a discrete-time system with a total number of $T$ time periods.     
    The candidate SeV set at time period $t$ is denoted as $\mathcal{N}_t$, which may change over time. Assume that the density of SeV is high enough such that $\mathcal{N}_t \neq \emptyset$ for $\forall t$, otherwise the TaV may seek help from cloud servers at RSUs or in the Internet.
    The computation task that requires to be offloaded at time period $t$ is modeled by three parameters $(x_t,y_t,w_t)$ according to \cite{mao2017mobile}: $x_t$ is the input data size (in bits) to be transmitted from TaV to SeV, $y_t$ is the output data size (in bits) which is the computation result transmitted back from SeV to TaV, and  $w_t$ is the computation intensity (in CPU cycles per bit) representing how many CPU cycles are required to process one bit input data. The total required CPU cycles to execute the task is then given by $x_tw_t$.
    
    There are three offloading steps:
    
    \textbf{Task upload:} 
    During time period $t$, the TaV selects a subset $\mathcal{S}_t$ of candidate SeVs with $\mathcal{S}_t \subset \mathcal{N}_t$, and offloads the task replica to them simultaneously. We assume that the number of selected SeVs is fixed as $K$, i.e., $|\mathcal{S}_t|=K\leq |\mathcal{N}_t|$. For each SeV $n \in \mathcal{N}_t$, denote the uplink wireless channel state as $h^{(u)}_{t,n}$, and the interference power $I^{(u)}_{t,n}$. The channel bandwidth and transmission power are fixed as $W$ and $P$, and the noise power is $\sigma^2$. Then the uplink transmission rate $r^{(u)}_{t,n}$ between the TaV and SeV $n$ can be written as
     \begin{align} \label{uplink_rate}
    r^{(u)}_{t,n} = W\log_2\left(1 + \frac{Ph^{(u)}_{t,n}}{\sigma^2 + I^{(u)}_{t,n}}\right).
    \end{align}
    Thus the delay of uploading the task from TaV to SeV $n$ is      
   \begin{align}  
   d_t^{(u)}(t,n) = \frac{x_t }{r^{(u)}_{t,n}}.
   \end{align}
     
    \textbf{Task execution:} 
    After receiving the input data from TaV, each SeV $n \in \mathcal{S}_t$ executes the task independently. Denote the maximum CPU frequency of SeV $n$ as $C_n$ (in CPU cycles per second). Each SeV may process multiple computation tasks at the same time, either from its own user equipments or other TaVs, and the allocated CPU frequency for the TaV of interest is $f_{t,n}\in [0,C_n]$. Then the computation delay of SeV $n$ in time period $t$ is 
    \begin{align}
    d_c(t,n) = \frac{x_t w_t}{f_{t,n}}.
    \end{align}
    
    \textbf{Result feedback:} 
    The computation result is finally transmitted back from each selected SeV $n \in \mathcal{S}_t$ to the TaV. Similar to \eqref{uplink_rate}, the downlink transmission rate between SeV $n$ and TaV in time period $t$ is given by 
     \begin{align} \label{downlink_rate}
    r^{(d)}_{t,n} = W\log_2\left(1 + \frac{Ph^{(d)}_{t,n}}{\sigma^2 + I^{(d)}_{t,n}}\right),
    \end{align}
    where $h^{(d)}_{t,n}$ and $I^{(d)}_{t}$ are the downlink channel state and the interference at the TaV respectively.   
    Therefore the downlink delay from SeV $n$ to TaV is      
    \begin{align}  
    d_t^{(d)}(t,n) = \frac{y_t }{r^{(d)}_{t,n}}.
    \end{align}
    
    As a result, the offloading delay $d(t,n)$ of SeV $n$ is the sum of uplink and downlink transmission delay and the computation delay, written as
    \begin{align}  
    d(t,n) = d_t^{(u)}(t,n)+d_c(t,n) +d_t^{(d)}(t,n).
    \end{align}
        
    The actual offloading delay of each task that the TaV experiences only depends on 
    \begin{align}
    	\min_{n\in \mathcal{S}_t}d(t,n).
    \end{align}
    However, we still require \emph{all} the other SeVs to finish execution and transmit the result, in order to record the offloading delay for the learning purposes, which will be introduced in detail in Section \ref{algo}.

    \subsection{Problem Formulation} \label{pro}
	Given a total number of $T$ time periods, our objective is to minimize the average offloading delay of tasks, by deciding which subset of SeVs should be selected to serve each task. The problem is formulated as:
	 \begin{align} \label{obj}
	\textbf{P1:}~\min_{\mathcal{S}_1,...,\mathcal{S}_T} \frac{1}{T}\sum_{t=1}^{T}\min_{n \in \mathcal{S}_t}d(t,n).  
	\end{align}
	
	The delay performance of each SeV mainly depends on three variables: uplink transmission rate $r^{(u)}_{t,n}$, allocated CPU frequency $f_{t,n}$, and downlink transmission rate $r^{(d)}_{t,n}$. If these variables are known to the TaV before it offloads each task, the TaV can then calculate the exact offloading delay $d(t,n)$ of each SeV $n$, and select single SeV $\arg\min_{n \in \mathcal{N}_t}d(t,n)$.
	However, due to the movements of vehicles, the transmission rates $r^{(u)}_{t,n}$ and $r^{(d)}_{t,n}$ are fast varying and hard to predict.
	Meanwhile, the allocated CPU frequency $f_{t,n}$ is not easy to know in prior due to the varying computation loads of SeVs.
	Thus TaV may lack the exact global state information, and can not distinguish which SeV provides the fastest computation for each task.
	
	Our solution is \emph{learning while offloading}: we let the TaV learn the delay performance of candidate SeVs through delay observations while offloading tasks. To be specific, till time period $t$, the TaV gets delay records $\{d(1,n),n \in \mathcal{S}_1\}$, ..., $\{d(t-1,n),n \in \mathcal{S}_{t-1}\}$, estimates the delay performance at the current time period, and selects a subset  $\mathcal{S}_t=\arg\min_{\mathcal{S} \in \mathcal{N}_t}\min_{n \in \mathcal{S}}d(t,n)$ to offload the task replica.
	

	\section{Learning while Offloading}\label{algo}
	
	In this section, we design learning-based task replication algorithm, which guides the TaV to learn the delay performance of candidate SeVs while offloading tasks, in order to minimize the average offloading delay.
	We consider a simplified scenario by assuming that tasks are of equal input, output data size $x_t=x_0$, $y_t=y_0$ and computation intensity $w_t=w_0$ for $\forall t$.
	In fact, tasks often have similar input and output data size ratio and computation intensity if they are generated by the same kind of applications. And tasks with diverse input data size can be partitioned into several subtasks of the same input data size and offloaded in sequence, e.g., a long video frame for object detection or classification can be divided into short video clips through video segmentation \cite{cvpr2010}.
	
	The task replication is an online sequential decision making problem, which have been investigated under the MAB framework.
	In classical MAB problem \cite{auer2002finite}, there are a fixed number of base arms with unknown loss distributions. In each time period, a decision maker tries a candidate base arm, observes its loss, and update the estimates of its loss distribution. The objective is to minimize the cumulative loss over time. 
	The classical MAB problem has been further extended to the CMAB problem \cite{chen2013cmab}, where in each time period the decision maker can try a subset of base arms (defined as a super arm), observe the loss of all the base arms composing this super arm, and minimize the cumulative loss of the system. 
	
	Our problem is similar to the CMAB problem with non-linear loss function: each candidate SeV corresponds to a base arm with unknown delay distribution, while the TaV is the decision maker who selects a subset $\mathcal{S}_t$ of SeVs in each time period $t$. 
	The TaV can observe the delay $ d(t,n)$ (loss) of all selected SeVs, and the system loss, i.e., the offloading delay, is the minimum of the observed delay $\min_{n \in \mathcal{S}_t}  d(t,n)$, which is a non-linear function.
	
	The major difference between our task replication problem and the existing CMAB problem is that, the candidate SeV set $\mathcal{N}_t$ may change over time since vehicles are moving, and it is difficult to predict when SeVs may appear or disappear, and how long they can act as candidates. 	
	How to efficiently learn the delay performance of candidate SeVs under such a dynamic environment has not been investigated in the existing work of CMAB.
	
	We thus take into consideration the time varying feature of candidate SeVs, and revise the existing CMAB algorithm in \cite{chen2016cmab} into learning-based task replication algorithm (LTRA), as shown in Algorithm 1.
	Let $\tilde{d}(t,n)=\frac{d(t,n)}{d_{max}}$ be the normalized delay, where $d(t,n), \forall n\in \mathcal{S}_t$ is the delay observed by TaV, and $d_{max}$ is the maximum delay allowed of each task offloading. If in time period $t$, the computation result from SeV $n$ is not successfully received by the TaV till $d_{max}$, we regard that the task is failed by SeV $n$, and set the observed delay $d(t,n)=d_{max}$ for learning purpose. And thus $\tilde{d}(t,n) \in [0,1]$.
	Denote $\hat D_n$ as the empirical distribution of the normalized delay $\tilde{d}(t,n)$ of SeV $n$, and $\hat F_n$ the cumulative distribution function (CDF) of $\hat D_n$.
	Let $k_{t,n}$ be the number of tasks offloaded to SeV $n$ so far, $\beta$ a constant factor, and $t_n$ the occurrence time of SeV $n$.


	\begin{algorithm}
		\caption{Learning-based Task Replication Algorithm}
		\begin{algorithmic}[1]
			\For {$t=1,...,T$}
			\If { Any SeV $n \in \mathcal{N}_t$ has not connected to TaV}
			\State Connect to any subset $\mathcal{S}_t \in \mathcal{N}_t$ once, with $n\in \mathcal{S}_t$.
			\State Update empirical CDF $\hat F_n$ of normalized delay $\tilde{d}(t,n)$ and selected times $k_{t,n}$ for each $n\in \mathcal{S}_t$. 
			\Else
			\State For each $n\in\mathcal{N}_t$, define CDF $\underline{G}_n$ as
			\begin{align} \label{utility}
			\underline{G}_n(x)=   \begin{cases} 0&x=0, \\
			\min \left\{ \hat{F}_n(x)+\sqrt{\frac{\beta\ln (t-t_n)}{k_{t-1,n}}},1 \right\}&0<x\leq1\end{cases} .          
			\end{align}		
			\State Select a subset of candidate SeVs, such that
			\begin{align} \label{bestset}
				\mathcal{S}_t= \argmin_{\mathcal{S} \subset \mathcal{N}_t } \mathbb{E}_{\underline{D}}\left[\min_{n \in \mathcal{S}}d(t,n)\right], 
			\end{align}
			where $\underline{D}_n$ is the distribution of CDF $\underline{G}_n$, and $\underline{D}=\underline{D}_1 \times \underline{D}_2 \times ... \times \underline{D}_{|\mathcal{N}_t|} $. 
			\State Offload the task replica to all the SeV $\forall n \in \mathcal{S}_t$.
			\State Observe delay $d(t,n)$ for each $n\in\mathcal{S}_t$.
			\State Update $\hat F_n$ and $k_{t,n}\leftarrow k_{t-1,n}+1$ for each $n\in \mathcal{S}_t$.
			\EndIf
			\EndFor		
		\end{algorithmic}
	\end{algorithm}

	In Algorithm 1, Lines 2-4 are the initialization phase, during which the TaV selects a subset of $K$ SeVs that contains at least one newly appeared SeV. Note that the initialization phase not only happens at the beginning of task offloading, but whenever new SeVs occur. 
	
	Lines 6-10 are the main learning phase. 		
	Due to the non-linearity of the offloading delay $\min_{n \in \mathcal{S}_t}  d(t,n)$, the offloading decision $\mathcal{S}_t$ depends on the \emph{entire delay distribution} of each candidate SeVs, rather than their means. Thus the learning algorithm keeps updating the empirical CDF $\hat F_n$ to learn the entire distribution, and makes offloading decisions according to $\hat F_n$.	
	In Line 6, the CDF $\underline{G}_n(x)$ defined in \eqref{utility} is a numerical upper confidence bound on the real delay CDF of each SeV $n$, which can balance the exploration-exploitation tradeoff during the learning process:
	The TaV tends to \emph{explore} SeVs with fewer selected times $k_{t,n}$ to learn good estimates of their delay distributions, while at the same time to \emph{exploit} SeVs with better empirical delay performance to optimize the instantaneous offloading delay.
	The padding term $\sqrt{\frac{\beta\ln (t-t_n)}{k_{t-1,n}}}$ also considers the occurrence time $t_n$ of each SeV $n$, such that the newly appeared SeVs can be well explored, while the empirical information of the existing SeVs can be exploited.
	
	In Line 7, the TaV selects a subset of candidate SeVs that minimizes the expectation of offloading delay according to \eqref{bestset}, where $\underline{D}_n$ is the distribution of CDF $\underline{G}_n$, and $\underline{D}=\underline{D}_1 \times \underline{D}_2 \times ... \times \underline{D}_{|\mathcal{N}_t|} $ is the joint distribution of all candidate SeVs. Calculating $\mathcal{S}_t$ is actually a minimum element problem, which can be solved by greedy algorithms \cite{goel2010how}. Then the TaV offloads the task replica to all the selected SeVs $n \in \mathcal{S}_t$, waits for their feedbacks to observe the delay, and finally updates the empirical CDF $\hat F_n$ of normalized delay $\tilde{d}(t,n)$ and selected times $k_{t,n}$.
	
	\subsection{Implementation Considerations}
	
	Since the offloading delay is continuous, LTRA may suffer from large storage usage and computational complexity as $t$ grows. To be specific, the observed delay values $d(t,n)$ of each SeV $n$ might be different in each time period, and thus the required storage for each empirical CDF $\hat F_n$ is $O(t)$. Meanwhile, it takes $O(t)$ time to calculate the numerical upper confidence bound $\underline{G}_n(x)$, and the minimum element problem in \eqref{bestset} is more complex to solve. 
	To reduce the storage usage and computational complexity of the algorithm, we can discretize the empirical CDF $\hat F_n$ to be $\tilde F_n$, by partitioning the range $[0,1]$ into $l$ segments with equal interval $\frac{1}{l}$. 
	The support of the discretized CDF $\tilde F_n$ is $\{\frac{1}{l},\frac{2}{l},...,1\}$, and if the normalized delay $\tilde{d}(t,n)$ belongs to $\left(\frac{j-1}{l},\frac{j}{l}\right]$, the delay used to update $\tilde F_n$ is $\frac{j}{l}$.
	
	\section{Algorithm Performance}\label{per}
	In this section, we characterize the performance of the proposed LTRA.
	To carry out theoretical analysis, we assume that the candidate SeV set $\mathcal{N}_t$ remains constant as $\mathcal{N}$ for the total $T$ time periods, and the delay $d(t,n)$ of each candidate SeV $n$ is independent from other SeVs and i.i.d across time.
	We will prove later through simulations in Section \ref{sim} that without the aforementioned two assumptions, LTRA can still work well.
	
	Let $\d_t=(d(t,1),...,d(t,N))$ be the delay vector of all candidate SeVs in time period $t$ with $N= |\mathcal{N}|$, and $L(\d_t, \mathcal{S}_t)=\min_{n\in \mathcal{S}_t}d(t,n)$ the loss function.
	The expected loss $\mu_{\mathcal{S}_t}$ for choosing a subset $\mathcal{S}_t$ of candidate SeVs does not change over time due to the i.i.d assumption, thus we omit the subscript $t$ and let $\mu_\mathcal{S}=\mathbb{E}[L(\d, \mathcal{S})]$.
	Moreover, let  $\mathcal{S}^*=\argmin_{\mathcal{S} \subset \mathcal{N} }\mu_\mathcal{S}$ be the optimal subset of SeV, and $\mu_{\mathcal{S}^*}=\min_{\mathcal{S} \subset \mathcal{N} }\mu_\mathcal{S}$ its expected loss.

	Define the cumulative learning regret $R_T$ as
	\begin{align}
		R_T=\mathbb{E}\left[   \sum_{t=1}^{T}L(\d_t, \mathcal{S}_t)   \right]   -T\mu_{\mathcal{S}^*},
	\end{align}
	which is the expected loss brought with learning as compared to the optimal decisions, since the TaV does not know which candidate SeV performs the best.
	
	For any suboptimal subset $\mathcal{S}$, let the expected delay gap $\Delta_\mathcal{S}=(\mu_\mathcal{S}-\mu_{\mathcal{S}^*})/d_{max}$. Define
	\begin{align}
	\Delta_n=\min\{\Delta_\mathcal{S}| \mathcal{S}\subset \mathcal{N}, n\in  \mathcal{S},  \mu_\mathcal{S}>\mu_{\mathcal{S}^*}    \},
	\end{align}
	and let $\mathcal{N}_s$ be the set of candidate SeVs which is contained in at least one suboptimal subset.
	
	In the following theorem, we provide an upper bound of the cumulative learning regret of LTRA.
	
	\begin{theorem} \label{theo}
		Let $\beta=\frac{2}{3}$, then $R_T$ is upper bounded by:
		\begin{align} \label{regret_bound}
		R_{T}\leq d_{max}\left(C_1 K \sum_{n\in\mathcal{N}_s}\frac{\ln T}{\Delta_n}+C_2 \right),
		\end{align}
		where $C_1=2136$ and $C_2=\left(\frac{\pi^2}{3}+1\right)N$ are two constants.
	\end{theorem}
	\begin{proof}
		See Appendix \ref{proof1}.
	\end{proof}
	
	Theorem 1 shows that, our proposed LTRA can provide a delay performance with bounded regret, compared to the genie-aided case, where the delay distributions of candidate SeVs are known in prior. To be specific, the cumulative learning regret grows logarithmically with $T$, and is also related to the number of selected SeVs and the performance gap between different subsets of SeVs.

	\section{Simulations} \label{sim}
	In this section, we carry out simulations to evaluate the delay performance and service reliability of the proposed LTRA.
	We first use SUMO\footnote{http://www.sumo.dlr.de/userdoc/SUMO.html} to simulate the traffic, and then import the floating car data generated by SUMO into MATLAB to evaluate the performance of LTRA. 
	
	The road used for traffic simulation is a $12\mathrm{km}$ segment of G6 Highway in Beijing, with two lanes and two ramps,  downloaded from Open Street Map (OSM)\footnote{http://www.openstreetmap.org/}. Vehicles come from either the start of the road or the ramps, and when a vehicle meets a ramp, it leaves the highway with a probability of $0.5$. The maximum speed allowed of both SeVs and TaVs is $20\mathrm{m/s}$. The arrival rate of SeV ranges from $0.05\mathrm{s^{-1}}$ to $0.4\mathrm{s^{-1}}$, and the arrival rate of TaV ranges from $0.01\mathrm{s^{-1}}$ to $0.2\mathrm{s^{-1}}$. 
	
	The floating car data of SUMO includes the type, ID, position, speed and angle of each vehicle, so that we can calculate the distance of each SeV and TaV in MATLAB. The communication range is set to $300\mathrm{m}$, and the wireless channel $h^{(u)}_{t,n}=h^{(d)}_{t,n}=A_0d_{st}^{-2}$, with $A_0=-17.8\mathrm{dB}$ and $d_{st}$ the distance between TaV and SeV, according to \cite{abdulla2016vehicle}. 
	The channel bandwidth $W=10\mathrm{MHz}$, transmission power $P=0.1\mathrm{W}$, and noise power $\sigma^2=10^{-13}\mathrm{W}$.
	
	For each task, we set the input data size $x_0=1\mathrm{Mbits}$, its computation intensity $\omega_0=1000 \mathrm{Cycles/bit}$, and the output data size is very small and omitted.
	The maximum CPU frequency $C_n$ of each SeV is uniformly chosen within $[2,8]\mathrm{GHz}$. 
	In each time period, the allocated CPU frequency for TaVs is randomly distributed from $0$ to $60\%C_n$ (each SeV also needs to process tasks from its own driving system or UEs, so it can not allocate all the computing resources for TaVs).
	Note that each SeV may provide service for multiple TaVs at the same time, and in the simulation, tasks offloaded by TaVs are served by the first-come-first-serve queue discipline. 
	Moreover, parameter $\beta$ in \eqref{utility} is $0.6$, and the default number of discretization segments is $50$.

	We compare our proposed LTRA with 3 other algorithms: 
	1) \textbf{Genie-aided Policy}: assume that the TaV knows the global state information of all candidate SeVs, and always selects single SeV with minimum delay. This policy can not be realized in the realistic VEC system.
	2) \textbf{Random Policy}: a naive policy, in which TaV randomly selects a SeV in each time period to offload the task.
	3) \textbf{Single Offloading}: a learning-based task offloading policy proposed in our previous work \cite{Sun2018ICC}, in which each TaV only selects a \emph{single} SeV for each task.

	\begin{figure}[!t]
		\centering	
		\includegraphics[width=0.45\textwidth]{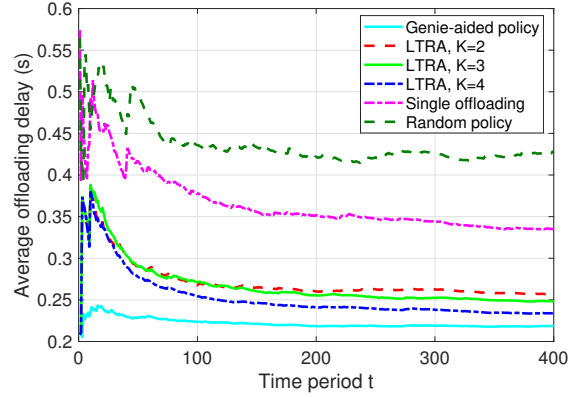}	
		\caption{Average offloading delay of LTRA.}
		\label{fig_delay}
	\end{figure}

	Fig. \ref{fig_delay} shows the average offloading delay of LTRA with the number of task replicas $K=2,3,4$. 
	On average, there are about $8$ candidate SeVs and $1$ other TaV around our target TaV.
	It is shown that with task replication, the delay performance is improved to about $0.25\mathrm{s}$, while learning-based single offloading algorithm can only achieve $0.33\mathrm{s}$. With increasing number of task replicas, the average offloading delay is closer to the genie-aided policy.

	\begin{figure}[!t]
	\centering	
	\subfigure[Average offloading delay.]{\label{fig_sev_delay}			
		\includegraphics[width=0.45\textwidth]{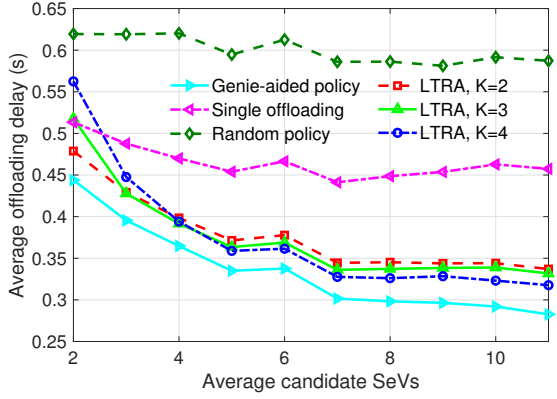}}		
	\subfigure[Task completion ratio.]{\label{fig_sev_ratio}	
		\includegraphics[width=0.45\textwidth]{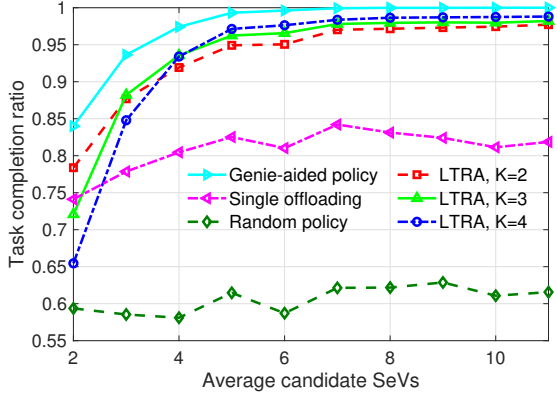}}
	\caption{Performance of LTRA under different number of candidate SeVs.}
	\label{fig_sev}
	\end{figure}

	Fig. \ref{fig_sev} shows the delay performance and service reliability under different SeV densities. 
	The x-coordinate, ranges from $2$ to $11$, is the average number of candidate SeVs around our target TaV, and still there is another $1$ TaV around.
	In Fig. \ref{fig_sev_delay}, the average offloading delay of LTRA decreases along with the increasement of candidate SeVs, since LTRA can exploit the redundant computing resources through  task replication.
	However, more task replicas do not always bring performance improvement.
	When computing resources are insufficient, too many task replicas may lead to long task queues at candidate SeVs, which is not efficient.
	Fig. \ref{fig_sev_ratio} shows the task completion ratio given deadline $0.6\mathrm{s}$.
	When there are more than $3$ candidate SeVs, the task completion ratio of LTRA outperforms single offloading.
	And when there are more than $7$ candidate SeVs, the service reliability of LTRA can reach over $98\%$, while single offloading only achieves about $80\%$. 
	Therefore, with sufficient computing resources in the VEC system, task replication is a promising method to enhance the reliability of computing services .
 
	\begin{figure}[!t]
		\centering	
		\includegraphics[width=0.45\textwidth]{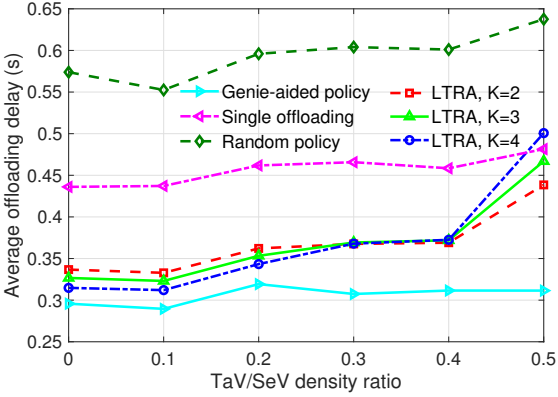}	
		\caption{Average offloading delay of LTRA under different TaV and SeV density ratio.}
		\label{fig_tav}
	\end{figure}

	In Fig. \ref{fig_tav}, the density of SeV is fixed, with about $8$ candidates around the target TaV.
	As more TaVs competing the computing resources with each other, the average offloading delay of LTRA increases. 
    To be specific, when the density ratio of TaV and SeV is below $0.3$, LTRA with $4$ task replicas outperforms LTRA with only $2$ replicas. However, as the density ratio grows higher, fewer number of replicas achieves better delay performance.
    Thus the number of task replica should be carefully selected under different traffic conditions.

		\begin{figure}[!t]
		\centering	
		\subfigure[Average offloading delay.]{\label{fig_step_delay}			
			\includegraphics[width=0.43\textwidth]{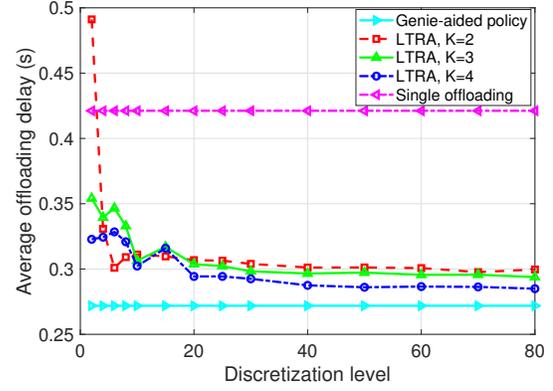}}		
		\subfigure[Runtime of LTRA.]{\label{fig_step_runtime}	
			\includegraphics[width=0.43\textwidth]{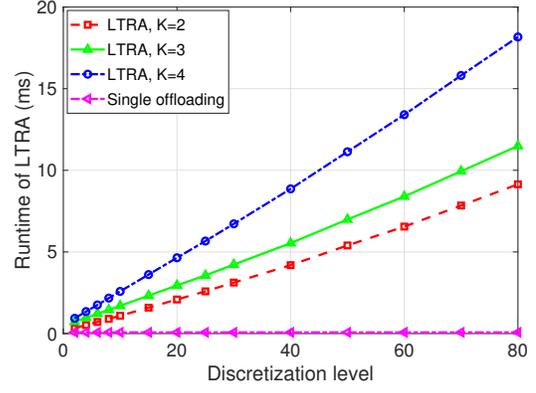}}
		\caption{Impact of discretization level $l$.}
		\label{fig_step}
	\end{figure}


	Finally, we explore the impact of discretization level $l$ on the average offloading delay and the runtime of the algorithm.
	When the delay distribution is discretized into very few segments, the runtime of LTRA is low, but the average offloading delay is very poor and fluctuates severely. When the discretization level is too high, the delay performance does not improve much, but it takes more time to run LTRA.
	To get good estimates of the realistic delay distributions,  while saving the runtime of LTRA at the same time, the discretization level should be carefully selected. For example, under our settings, the discretization level $l$ should be about $40$ to $50$.

	\section{Conclusions} \label{con}
	In this work, we have investigated the task offloading problem in VEC system, and proposed LTRA by combining the task replication and sequential learning techniques, in order to minimize the average offloading delay.
	LTRA enables each TaV to learn the delay performance of candidate SeVs while offloading tasks, and can adapt to the highly dynamic vehicular environment.
	We have carried out simulations under a realistic highway scenario, and compared the delay performance and service reliability of LTRA to the existing single offloading algorithm.
	Simulation results have shown that the average delay of LTRA is close to the optimal genie-aided policy and better than the single offloading policy. And when there are sufficient SeVs, the performance can be highly improved through a small number of task replications.
	Specifically, with a given deadline 0.6s, the task completion ratio of LTRA can reach $98\%$ with only two replicas, while single offloading can only achieve about $80\%$.

\section*{Acknowledgment}
This work is sponsored in part by the Nature Science Foundation of China (No. 91638204, No. 61571265, No. 61621091), and Intel Collaborative Research Institute for Mobile Networking and Computing.

\appendices{}

\section{Proof of Theorem 1} \label{proof1}
We prove that, under the assumption that the number of candidate SeVs is fixed, our delay minimization problem is equivalent to the reward maximization problem of standard CMAB investigated in \cite{chen2016cmab}, and the proposed algorithm LTRA is equivalent to the stochastically dominant confidence bound (SDCB) algorithm proposed in \cite{chen2016cmab}.

First, the objective functions are equivalent, since
\begin{align} 
&\min_{\mathcal{S}_1,...,\mathcal{S}_T} \frac{1}{T}\sum_{t=1}^{T}\min_{n \in \mathcal{S}_t}d(t,n)\nonumber\\
=&d_{max}\min_{\mathcal{S}_1,...,\mathcal{S}_T} \frac{1}{T}\sum_{t=1}^{T}\min_{n \in \mathcal{S}_t}\tilde{d}(t,n)\nonumber\\
\Leftrightarrow & \max_{\mathcal{S}_1,...,\mathcal{S}_T} \frac{1}{T}\sum_{t=1}^{T}\left[\max_{n \in \mathcal{S}_t}\left(1-\tilde{d}(t,n)\right)\right].
\end{align}
Since $\tilde{d}(t,n)\in [0,1]$, the reward function $R(\d_t,\mathcal{S}_t)=\max_{n \in \mathcal{S}_t}\left(1-\tilde{d}(t,n)\right)\in [0,1]$, satisfying assumption 2 in \cite{chen2016cmab} with upper bound $M=1$. Also, $R(\d_t,\mathcal{S}_t)$ is monotone, which satisfies assumption 3 in \cite{chen2016cmab}.

Second, the numerical upper confidence bound $\underline{G}_n(x)$ can be transformed to CDF $\underline{F}_n(x)$ defined in the SDCB algorithm in \cite{chen2016cmab}. Define $\hat{F}_n(x)$ as the CDF of $\tilde{d}(t,n)$, and $\hat{F}'_n(x)$ the CDF of $1-\tilde{d}(t,n)$.
It is easy to see that $\hat{F}(x)=1-\hat{F}'(1-x)$. Thus
\begin{align}
	&\underline{G}_n(x)=1-\underline{F}_n(1-x) \nonumber\\
	&=1- \begin{cases} \max \left\{ \hat{F}'_n(1-x)-\sqrt{\frac{\beta\ln t}{k_{t-1,n}}},0 \right\}&0\leq1-x<1, \\
	1&1-x=1\end{cases} \nonumber\\
	&=1- \begin{cases} \max \left\{ 1-\hat{F}_n(x)-\sqrt{\frac{\beta\ln t}{k_{t-1,n}}},0 \right\}&0<x\leq1, \\
	1&x=0\end{cases} \nonumber\\
	&= \begin{cases} 0&x=0, \\
	\min \left\{ \hat{F}_n(x)+\sqrt{\frac{\beta\ln t}{k_{t-1,n}}},1 \right\}&0<x\leq1\end{cases} .
\end{align}

By substituting the reward upper bound $M=1$, and let $\alpha=1$ in Theorem 1 in \cite{chen2016cmab}, \eqref{regret_bound} can be derived.

	%

\end{document}